\documentclass[reqno,10pt]{amsart}
\usepackage{amsmath,amsfonts,amsgen,amstext,amsbsy,amsopn,amsthm, amssymb}
\usepackage{bbm}
\usepackage{bbold}
\usepackage{enumitem}
\usepackage{verbatim}

\newtheorem{lemma}{Lemma} \newtheorem{theorem}{Theorem}
\newtheorem{Theorem}{Theorem}
\newtheorem{proposition}{Proposition} 
\newtheorem{corollary}{Corollary} \theoremstyle{definition}

\theoremstyle{remark}
\newtheorem{remark}[theorem]{Remark}

\newcommand\nn\nonumber
\newcommand{\1}{\mathbb 1}
\newcommand{\M}{\mathcal M}

 \newcommand{\R}{\mathbb{R}}

 \newcommand{\F}{\mathcal{F}}
\newcommand{\E}{\mathcal{E}}

\newcommand\LM{L^2 (\mathbb R ^{3N} )}

\newcommand{\norm}[1]{\left\| #1 \right\|}
\renewcommand{\d}{\mathrm{d}}

\begin{document}

\title{The ground state energy of the multi-polaron in the strong coupling limit}

\author[I. Anapolitanos]{Ioannis Anapolitanos}\address{Universit\"at Stuttgart, Fachbereich Mathematik,
70550 Stuttgart, Germany} \email{anapoli@mathematik.uni-stuttgart.de}

\author[B. Landon]{Benjamin Landon}\address{Department of Mathematics and Statistics, McGill
  University, 805 Sherbrooke Street West, Montreal, QC H3A 2K6,
  Canada} \email{benjamin.landon@mail.mcgill.ca}

\date{\today}

\begin{abstract}
We consider the Fr\"ohlich $N$-polaron Hamiltonian  in the strong
coupling limit and bound the ground state energy from
below. In particular, our lower bound confirms that the ground state
energy of the Fr\"ohlich polaron and the
ground state energy of the associated Pekar-Tomasevich variational
problem are asymptotically equal in the strong coupling limit. We generalize the operator approach that was used to prove a similar result in the $N=1$
case in \cite{LT} and apply a Feynman-Kac formula to obtain the same result for an arbitrary particle
number $N \geq 1$.
\end{abstract}

\thanks{\copyright\, 2012 by the authors. This paper may be reproduced, in its entirety, for non-commercial purposes.}

\maketitle

\section{Introduction and Main Results}
\setdescription{font=\normalfont}
The polaron was modelled in 1937 by Fr\"ohlich  in order to explain the phenomenon of electrical breakdown in crystals \cite{Fr}. The model describes a single electron interacting with the phonon field of the lattice of atoms of a polar crystal. In the time since Fr\"ohlich's initial studies, the polaron has found interest not only in condensed matter theory, but also as a simple example of non-relativistic quantum field theory.  `Polaron' refers to the quasiparticle of the electron together with the phonons created by its interaction with the polar crystal. We refer the reader to \cite{AD} and \cite{M} for more details.  

The $N$-polaron is a generalization of the single electron case, and describes $N$ electrons interacting with the phonon field and also each other through the Coulomb force.  There is an asymptotic formula for the ground state energy in the strong coupling limit. In the case $N=1$ it was proposed by Pekar in \cite{P} and proven by Donsker and Varadhan in \cite{DV}. We are interested in proving this formula for polarons of an arbitrary electron number, i.e., when $N >1$.

The Fr\"ohlich Hamiltonian describing a single electron interacting with a polar crystal is given by
\begin{align}
H^{(1)} =p^2 - \sqrt{\alpha} \phi(x) + H_f ,
\end{align}
and acts on the Hilbert space $L^2 (\mathbb{R}^3 ) \otimes \mathcal{F}$, with
$\mathcal{F}$ the bosonic Fock space over $L^2(\mathbb{R}^3)$. Here $p = - i \nabla$ is the
electron momemtum operator and the phonon field energy is
\begin{align}
H_f = \int_{\mathbb{R}^3} \d k a^* (k) a(k)
\end{align}
where $a^*(k)$/$ a(k)$ are the creation/annihilation operators for a
phonon of momentum $k$. The interaction of the crystal modes with
the electron is
\begin{align}
\phi (x) = \frac{1}{\sqrt{2} \pi} \int_{\mathbb{R}^3} \frac{\d k}{| k
|} \left( e^{ikx} a(k) + e^{-ikx} a^* (k) \right).
\end{align}
The coupling constant $\alpha > 0$
describes the strength of the interaction between the electron and
the polar crystal. For a careful definition of $H^{(1)}$ as a self-adjoint operator that is bounded from below, we refer the reader to \cite{N} (see also \cite{MS}). For our methods, this definition is not very important as one can always interpret $H^{(1)}$ as a quadratic form that is bounded from below. We define its ground state energy by $E^{(1)} (\alpha)$, where we have denoted its dependence on $\alpha$.

 The Hamiltonian describing $N$ electrons in a polar crystal (the
`$N$-polaron' or `multipolaron') is given by
\begin{align}\label{eqn:nham}
H^{(N)} _U = \sum_{i=1} ^N \left( p_i ^2  - \sqrt{\alpha} \phi ( x_i
) \right) + H_f + U V_C (X) ,
\end{align}
with $X = (x_1, ..., x_N ) \in \mathbb{R}^{3N}$. Here $p_i = -i
\nabla _{x_i}$ is the momentum operator for the $i$-th electron and
\begin{align}
V_C (X) = \sum_{i< j} \frac{1}{| x_i - x_j |},
\end{align}
is the Coulomb repulsion between the electrons. The Hamiltonian  $H^{(N)}
_U $ acts on $L^2 ( \mathbb{R}^{3N} ) \otimes \mathcal{F}.$ 
Physically, the parameter $U$ is the square of the electron charge
and satisfies $U > 2 \alpha$. However, we will also consider the  regime $U \leq 2 \alpha$.

The ground state energy of $H^{(N)} _U$ is denoted by $E^{(N)}_U
(\alpha)$, that is
\begin{align}
E^{(N)}_U (\alpha) = \inf_{\norm{\psi} = 1 } \langle \psi \vert H_U
^{(N)} \vert \psi \rangle .
\end{align}
It is well known that $E^{(N)}_U (\alpha)>-\infty$ and this is even a simple consequence of our methods. In fact, when $U > 2 \alpha$, the system is stable; that is, the ground state energy is bounded below by a constant times the particle number \cite{FLST}. 

The strong coupling limit, in which $\alpha \to \infty$, was first studied by Pekar in the 1950s \cite{P}. He hypothesized that in this limit, the asymptotic ground state would equal a product state $\psi \otimes \xi$, for $\psi$ an electronic wave function and $\xi$ a phonon wave function. If one computes $\langle \psi \otimes \xi \vert H_U ^{(N)} \vert \psi \otimes \xi \rangle $ using the Pekar ansatz, it is easy to determine the minimizing phonon wave function $\xi$ for a given electronic wave function $\psi$. With this choice of $\xi$, one is led to the $N$-particle Pekar-Tomasevich (PT) functional defined for $\psi \in H^{1} (\mathbb R ^{3N} )$ as
\begin{align} \label{eqn:ptfunct}
\mathcal{P}^{(N)}_U [ \psi ] =\sum_{i=1}^N \int_{\mathbb{R}^{3N}} |
\nabla_{x_i} \psi | ^2 \d X + U \sum_{i<j} \int_{\mathbb{R}^{3N}}
\frac{ |\psi(X) |^2}{|x_i - x_j|} \d X - \alpha \int \int
_{\mathbb{R}^3 \times \mathbb{R}^3} \frac{\rho _{\psi} (x)
\rho_{\psi} (y) }{|x-y|} \d x \d y ,
\end{align}
where
\begin{align}
\rho_{\psi} (x) = \sum_{i=1}^N \int_{\mathbb{R}^{3(N-1)}} |
\psi(x_1, ..., x_{i-1}, x, x_{i+1}, ..., x_N ) |^2 dx_1 ...
\widehat{dx_i} ... dx_N .
\end{align}
The hat indicates that $\d x_i$ is omitted in the integration. We
define the ground state energy of the PT functional to be
\begin{align} \label{eqn:ptenergy}
\mathcal{E} ^{(N)}_U ( \alpha ) = \inf \left\{ \mathcal{P}^{(N)} _U
[ \psi ] : \psi \in H^1 ( \R^{3N} ), \norm{\psi}_{L^2} =
1  \right\}.
\end{align}
In the case $N=1$ we drop the subscript and write $\E^{(1)}(\alpha)$. This case was studied by Lieb \cite{L} who proved that there is a unique minimizer up to translations. The PT functional in the case $N>1$ is studied in \cite{Le}, where the existence of minimizers  is proven. Lewin also addresses the problem of binding in the PT functional. In the regime $U \leq 2 \alpha$ the ground state energy of the PT functional is studied in \cite{BB}; they also sketch the derivation of (\ref{eqn:ptfunct}) from the Pekar ansatz. 

A scaling argument gives
\begin{align}\label{eqn:scaling}
\mathcal{E}_U ^{(N)} ( \alpha ) = \mathcal{E}_{U/ \alpha } ^{(N)} (
1 ) \alpha^2.
\end{align}
Furthermore, because one can construct wave functions where the electrons are arbitrarily far apart from one another, the ground state energy of the PT functional is subadditive:
\begin{align}
\mathcal{E} ^{(N)} _U ( \alpha ) \leq \mathcal{E} ^{(k)} _U ( \alpha
)  + \mathcal{E} ^{(N-k)} _U ( \alpha ) .
\end{align}

 Pekar's hypothesis amounts to the conjecture that in the strong coupling limit, the asymptotic formula
\begin{align}
\lim_{\alpha \to \infty} \alpha ^{-2} E_U ^{(N)} (\alpha ) = \mathcal E ^{(N)} _\nu (1) , \label{eqn:limit}
\end{align}
holds.   Here, $\nu = U / \alpha$ is a fixed constant.

Because it arises using a variational ansatz, $\mathcal E_U ^{(N)} (\alpha )$ is automatically an upper bound for $E^{(N)}_U (\alpha) $. In 1983, Donsker and Varadhan first confirmed (\ref{eqn:limit}) in the case $N=1$ \cite{DV}. However, their proof did not give a rate of convergence and did not easily generalize to other settings. In 1995, Lieb and Thomas developed a simpler proof using operator methods \cite{LT}. They bound $E^{(1)} (\alpha )$ from below by $\mathcal E ^{(1)} (\alpha)$ 
minus an error term which is negligible in the strong coupling limit, thus proving (\ref{eqn:limit}).

In the paper \cite{MS}, Miyao and Spohn treated the strong coupling limit in the bipolaron case (their proof lies in the appendix) using the methodology of \cite{LT}. It was also claimed in \cite{FLST} that the proof in \cite{MS} generalizes to the $N > 2$ case. Unfortunately, the proof in \cite{MS} is slightly incomplete. While not confirming the formula (\ref{eqn:limit}), their proof \emph{does} give a lower bound on the energy of wave functions that describe electrons which are all localized to the same small box in space (the case that their methods do not address is when the electrons are very far apart from one another). Furthermore, this generalizes easily to $N > 2$. In fact, this generalization is a crucial element in our own proof.

It is our contribution to complete the proof using a Feynman-Kac formula, and to make the generalization to the $N > 2$ case. In particular, we prove a lower bound for the ground state energy $E^{(N)}_U (\alpha)$ which implies (\ref{eqn:limit}) for any $N$.

The astute reader will have noticed that we do not impose any symmetry constraints on the electron wave functions.
  Ideally, one would like to obtain a lower bound for the ground state energy $E^{(N)}_U (\alpha)$ (taken over only antisymmetric 
  wave functions) in terms of $\mathcal E ^{(N)} _U (\alpha)$ but with the infimum in (\ref{eqn:ptenergy}) replaced by an 
  infimum over the antisymmetric wave functions in $H^{1} ( \mathbb R ^{3N} )$. However, our methods do not directly generalize to this setting.

We now state our main result. In particular, it confirms (\ref{eqn:limit}). Recall that the ratio $\nu = U / \alpha$ is fixed.
\begin{Theorem}\label{main} 
\begin{enumerate}[label=(\roman{*}), ref=(\roman{*}),font=\normalfont]
\item For any given $ \nu>2$ and $\epsilon>0$ there exists a constant $C_{\epsilon,\nu}>0$
 such that  for all $ N \in \mathbb{N}$ and  $\alpha > \epsilon N^4 $,
\begin{align}
E_U ^{(N)} ( \alpha )\geq \alpha^2
\mathcal{E}^{(N)}_\nu (1) -C_{\epsilon,\nu} \alpha^{9/5} N^{9/5}.
\label{eqn:main1}
\end{align}
\item  For any given $\nu >0$ and $\epsilon>0$, there exists a
  constant $D_{\epsilon,\nu}>0$ such that for all $N \in \mathbb{N}$ and  $\alpha > \epsilon$,
\begin{align}
E_U^{(N)} ( \alpha )\geq \alpha^2 \mathcal{E}^{(N)}_\nu (1) -D_{\epsilon, \nu}
\alpha^{42/23} N^3. \label{eqn:main2}
\end{align}
\end{enumerate}
\end{Theorem}
The asymptotic formula (\ref{eqn:limit}) allows us to extract information about the Fr\"ohlich Hamiltonian from the PT approximation. In \cite{Le}, it was proven that for any $N$, there is a $\nu(N) >2 $ so that the binding inequality
\begin{align}
\E^{(N)}_\nu (1) < \E^{(N-k)}_\nu (1) + \E^{(k)}_\nu (1) ,
\end{align}
holds for any $\nu \leq \nu(N)$ and $k$. This clearly implies
\begin{corollary}
For any $N$ there is a $\nu(N) > 2$ so that for every $\nu < \nu(N)$, there is an $\alpha(N, \nu) >0$ so that the binding inequality
\begin{align}
E^{(N)}_U (\alpha) < E^{(N-k)}_U (\alpha) + E^{(k)}_U (\alpha) ,
\end{align}
holds for any $\alpha \geq \alpha(N, \nu)$ and any $k$.
\end{corollary}
In Theorem \ref{main} we have stated two lower bounds which are valid for different relative sizes of $\alpha$ and $N$; both imply (\ref{eqn:limit}). It is an eventual goal to prove a lower bound of the kind (when $\nu > 2 $ and hence the system is stable),
\begin{align}
E^{(N)}_U (\alpha ) \geq\alpha^2 \E^{(N)}_\nu (1) - O (\alpha^{9/5}, N). \label{eqn:conj}
\end{align}
Such a lower bound would imply that one can, in a certain sense, `commute' the two limits $N \to \infty$ and $\alpha \to \infty$ when considering the quantity
\begin{align}
\lim_{\substack{N \to \infty\\  \alpha \to \infty}} \frac{E^{(N)}_U (\alpha )}{N \alpha^2}.
\end{align}
In this sense, (\ref{eqn:main1}) is not optimal in the $N$ dependence of the error term. The $N^{9/5}$ term is close; however, as stated, the bound does not necessarily hold if one takes $N \to \infty$ first. We hope that (\ref{eqn:main1}) is a first step towards (\ref{eqn:conj}).

Let us comment on the restriction on $\alpha$  in Theorem 1(i). In the physical regime $\nu >2$ where the PT energy satisfies $\E^{(N)}_\nu (1) =O (N)$, the assumption $\alpha>\epsilon N^4$ is necessary for the error in \eqref{eqn:main1} to be small. What we mean is that, if one rewrites the RHS
  of (\ref{eqn:main1}) as $\alpha^2 N [ \mathcal{E}^{(N)}_\nu (1)/N - C_{\epsilon, \nu} \alpha^{-1/5} N^{4/5} ]$, then it is easy to see that the error term will be small relative to the PT energy only if $\alpha \gg N^4$. Therefore, (\ref{eqn:main1}) is only useful if the assumption $\alpha > \epsilon N^4$ is satisfied, and so this assumption isn't very restrictive anyway.

Comparatively, in the regime $\nu<2$, we have that $\E_U ^{(N)} ( \alpha ) =O( N^3) $ \cite{BB}. Rewriting the RHS of (\ref{eqn:main2}) as $N^3 \alpha^2 [\mathcal{E}^{(N)}_U (1) / N^3 - C_{\epsilon} \alpha^{-4/23} ]$, we see that the corrections to the PT energy are small regardless of the relative size of $N$ and $\alpha$ (as long as $\alpha \gg 1$).

\subsection{Outline of Proof}
 Our proof generalizes the operator methods developed in \cite{LT}, which treated the $N=1$ case, to polarons of arbitrary electron number. Using a continuous version of the IMS localization formula, we localize the electrons to cubes of a fixed side length in $\mathbb{R}^3$. While we  have localized each of the electrons to their own cube, we do not, a priori, know anything about how far the electrons are from one another. To quantify how spread out they are, we partition the electrons into disjoint clusters.
 Electrons in a single cluster are not too far from each other, and the electrons in different clusters are separated by some minimum distance.

The generalization of the methods in \cite{LT} (and also \cite{MS}) provides a lower bound on the energy of the electrons that depends on how far apart the electrons are from one another; that is, if they are very spread out, then the lower bound is not very useful. What we would like to do is ignore the interaction between electrons lying in different clusters, and apply the methodology of \cite{LT} to obtain a bound on the energy of each individual cluster. Since the electrons in a single cluster will be tightly packed, the bounds given by \cite{LT} will be good enough to handle the energy of the clusters.

In \cite{FLST}, the Feynman-Kac formula was applied in order to bound the interaction of two localized electrons sufficiently far from each other. We generalize this to the $N$-polaron case to  show that we can ignore the interaction between different clusters (at an appropriate error) and treat each cluster as its own subsystem. Subsequently, we sum up the bounds on each cluster energy obtained using \cite{LT}; this will in turn yield Theorem \ref{main}.

{\bf Acknowledgements:} The research of B.L. was partly supported by NSERC. The research of I.A. was supported by the 
 German Science Foundation (DFG). The authors wish to thank R. Seiringer for many useful discussions and comments on a draft of this paper. The authors  are grateful to M. Griesemer and D. Wellig for numerous stimulating discussions and access to a preprint of their paper \cite{GW}.

\section{Proof of Theorem 1}
{\bf Step 1} (Localization). The first step in proving Theorem \ref{main} is to apply a continuous version of the IMS localization
formula in order to localize each electron into a cube of finite
side length.

Fix $R > 0$ and for $a_i \in \mathbb R^3, i \in \{1,2,...,n\}$ let $C_R(a_1, ..., a_N) =
\left(-R/2, R/2\right)^{3N} +(a_1, ..., a_N )\subseteq
\mathbb{R}^{3N}$. $C_R (a_1, ..., a_N)$ is the $3N$-cube of side
length $R$ centered at $(a_1, ..., a_N)$. Let $\mathcal{B}_R$ be the
set of normalized wave functions for which the support of their
electronic part is contained in $C_R (a_i, ..., a_N)$ for
some $a_i \in \mathbb R ^3$ (i.e. $\psi \in L^2 (C_R (a) ) \otimes \mathcal{F}$ for some $a \in \mathbb R ^{3N}$). The following proposition is proven in \cite{LT} in the case $N=1$, where the IMS localization formula is applied.  We generalize this to arbitrary $N$.
\begin{proposition}\label{prop:localization}
For $R>0$,
\begin{align}\label{eqn:lowerb}
E^{(N)}_U (\alpha)  \geq \inf_{\psi \in \mathcal{B}_R} \langle \psi
\vert H^{(N)} _U \vert \psi \rangle - \frac{3 N \pi^2}{R^2} .
\end{align}
\end{proposition}

\begin{proof}
Let $\epsilon > 0$ and $\Phi$ be a normalized wave function such
that $\langle \Phi \vert H_U^{(N)} \vert \Phi \rangle \leq E^{(N)}_U
(\alpha) +\epsilon.$ For real-valued $\psi \in C^{\infty} _0 (C_R
(0) )$ set $\psi_Y (X) = \psi (X-Y).$ A direct calculation leads to
\begin{align}
\int_{\mathbb R ^{3N}} \mathrm{d}Y \langle \psi_Y \Phi \vert H_U^{(N)} \vert \psi_Y \Phi \rangle 
& = \int_{\mathbb R ^{3N}} \d Y | \psi (Y) |^2 \left( \langle \Phi \vert H_U^{(N)} \vert \Phi \rangle
\right) + \norm{ \nabla \psi }^2_{L^2 ( \mathbb{R}^{3N})}.
\end{align}

 We can take $\psi$ so that  $\norm{\nabla \psi }_{L^2}^2 \leq 3N \pi^2
/R^2 + \epsilon$ and $\norm{\psi}_{L^2} = 1$ (i.e., it approximates the ground state of the Dirichlet Laplacian). For this choice of $\psi$,
\begin{align}
\int_{\mathbb R ^{3N}} \d Y \left[ \langle \psi_Y \Phi \vert H_U^{(N)} \vert \psi _Y
\Phi \rangle - (E^{(N)}_U(\alpha) + 3N \pi^2 /R^2 + 2\epsilon )
\langle \psi_Y \Phi \vert \psi_Y \Phi \rangle \right] \leq 0.
\end{align}
There must be a $Y_0 \in \mathbb{R}^{3N}$ for which the integrand
appearing above is non-positive and $\langle \psi_{Y_0} \Phi \vert
\psi_{Y_0} \Phi \rangle$ is nonzero. Therefore,
\begin{align}
\langle \psi_{Y_0} \Phi \vert H_U^{(N)} \vert \psi_{Y_0} \Phi
\rangle / \langle \psi_{Y_0} \Phi \vert \psi_{Y_0} \Phi \rangle \leq
E_U^{(N)} (\alpha) + 3N \pi^2 /R^2 + 2\epsilon.
\end{align}
Since $\psi_{Y_0} \Phi ( \langle \psi_{Y_0}\Phi \vert \psi_{Y_0}\Phi
\rangle )^{-1/2} \in \mathcal{B}_R$ the claim is proven.
\end{proof}

We now fix a $\Phi \in \mathcal{B}_R$ for an $R>0$ which will
be chosen later. Each of the $N$ electrons described by $\Phi$ is
located in a cube $Q_i$ of side length $R$ centered at some $a_i \in
\mathbb{R}^3.$  The remainder of the proof constitutes finding a
lower bound for $\langle \Phi \vert H^{(N)}_U \vert \Phi \rangle $ which holds independently of $\Phi$; i.e., finding a lower bound for the RHS of (\ref{eqn:lowerb}).

{\bf Step 2} (Partitioning of electrons into clusters). The direct
generalization of \cite{LT} and \cite{MS} to the $N$-polaron case
produces a lower bound on $\langle \Phi \vert H_U ^N (\alpha ) \vert \Phi \rangle$ which depends on how
far apart the localized electrons are from one another. Herein lies
the main difficulty in our proof: we do not, a priori, know that the
cubes $\{Q_i \}_{i=1} ^N$ within which the electrons are contained
are close enough to each other so that this bound is useful. In
order to overcome this challenge, we must quantify how spread out
the electrons are. This is accomplished by partitioning the
electrons into clusters.

Let $d(i,j) =\text{dist} \{ Q_i, Q_j \}$ be the distance between the
cubes containing the supports of the $i$-th and $j$-th electrons. Let
$d>0$ be a fixed constant.  We define the relation $\sim$ on
$\{1,...,N\}$ by $j \sim k$ iff $\exists$ a sequence $( i_1, ..., i_M
) $ with $ i_1 = j$, $i_M = k$ satisfying $d (i_l , i_{l+1}) <d$ for
$1 \leq l < M$. Clearly $\sim$ is a relation of equivalence and we
let $G_1 , ..., G_l$ be the disjoint equivalence classes to which we
will give the suggestive name of `clusters.' The
interpretation of $\sim$ is that we have placed two electrons in
the same cluster iff either they are less than $d$ apart or they are
connected by a sequence in which consecutive electrons are less than
$d$ apart.

There are two important properties of these clusters:
\begin{description}
\item[(P1)] Let $| G_i | = N_i $ be the cardinality of the cluster $G_i$.
Then it is clear that the union of the supports of the electrons in
$G_i$ is contained in a cube in $\mathbb{R}^{3}$ of side length no
greater than $ N_i ( R + d ).$
\item[(P2)] If the electrons $i$ and $j$ are in two different clusters, then $d(i,j) \geq d.$
\end{description}

Property (P2) will allow us to apply the Feynman-Kac formula to bound the interaction between electrons in different clusters. Property (P1) will allow us to apply the methods of \cite{LT} to bound the remaining energy of each cluster of electrons.

{\bf Step 3} (Removing the inter-cluster interaction). 
For the cluster $G_i$, we define
\begin{align}
\widehat{E} (G_i) = \inf_{\varphi}\mbox{}'   \langle \varphi \vert
H^{(N_i)}_U \vert \varphi \rangle ,
\end{align}
where the infimum is taken over $N_i -$polaron wave functions
$\varphi$ which have the same support properties as the electrons in
$G_i$. More precisely, if the electrons in $G_i$ are localized in
the cubes $Q_{j_1}, ..., Q_{j_{N_i}}$ then the $\varphi$ appearing in the above infimum are also
supported in the same $N_i$ cubes (i.e., $\varphi \in L^2 (Q_{j_1} \times ... \times Q_{j_{N_i}} ) \otimes \mathcal F$). Physically, $\widehat{E} (G_i)$
may be interpreted as the lowest energy the electrons in the cluster
$G_i$ can have (ignoring, of course, the existence of the other
$N-N_i$ electrons).

 The following lemma provides a lower bound for $E_U ^{(N)} (\alpha)$ in terms of the cluster energies $\{ \widehat{E} (G_i) \}_i$ and an error term due to the inter-cluster interaction.
\begin{lemma}\label{lem:fk}
For our fixed $\Phi \in \mathcal{B}_R$ we have,
\begin{equation}\label{eqn:Feyngen}
\langle \Phi, H_U^{(N)} \Phi \rangle \geq \sum_{i} \widehat{E}(G_i)
-\sum_{k_1 \neq k_2} \sum_{\substack{ i \in G_{k_1} \\ j \in
G_{k_2}}} \left( \frac{\alpha}{d(i,j)}-\frac{U}{2d(i,j)+4
\sqrt{3}R}\right).
\end{equation}
In particular, we have,
\begin{align}
\langle \Phi, H_U^{(N)} \Phi \rangle \geq \sum_{i} \widehat{E}
(G_i)- \frac{\alpha N^2}{d} \label{eqn:Feyngencons}.
\end{align}
If $ \nu >2$ and $d \geq 4\sqrt{3} R/(\nu-2)$ then,
\begin{align}
\langle \Phi \vert H^{(N)}_U \vert \Phi \rangle \geq  \sum_{i}
\widehat{E} (G_i) \label{eqn:lem1}.
\end{align}
\end{lemma}
\begin{remark}
This is a direct generalization of Lemma 1 in \cite{FLST}, in which
the $N=2$ case is proven.
\end{remark}
\begin{proof}
The Feynman-Kac formula \cite{R} implies that the infimum of the LHS of
(\ref{eqn:Feyngen}) over normalized wave functions localized in $Q_1
\times ... \times Q_N$ equals\begin{align}
 -\lim_{T \to \infty} \frac{1}{T}\log Z_{Q_1, ..., Q_N } (T) ,
\end{align}
where
\begin{align}
Z _{Q_1 , ..., Q_N } (T) = \int _{Q_1} \d x_1 ... \int_{Q_N} \d x_N
\int \d W_{x_1} ^T
( \omega _1 ) ... \d W_{x_N} ^T (\omega_N ) \chi _{Q_1} (\omega _1 ) ... \chi _{Q_N} ( \omega _N )  \nonumber \\
\times \exp \left[ \alpha \int_{\mathbb{R}} \d s \frac{e^{- |s|}}{2} \sum_{i,
j = 1} ^N \int_0 ^T \frac{\d t}{| \omega_i (t) - \omega_j (t+s)|} - U
\sum_{i < j} \int_0 ^T \frac{\d t}{| \omega_i (t) - \omega_j (t)|}
\right].\label{eqn:feynz}
\end{align}
Above, $\d W_{x_j} ^T$ denotes the Wiener measure of closed Brownian
paths  $\omega_i :[ 0,T] \to \mathbb{R}^3$, starting and ending at $x_j$.
 $\chi_{Q_j} ( \omega_j) $ is a characteristic function which is $1$ if $\omega_j$ stays inside the
cube $Q_j$ for all times and $0$ otherwise.

We use our choice of clusters to bound the integrand in the above
equation.
 Specifically, we obtain for paths $\{ \omega_i \}_i$ lying in the cubes $\{ Q_i \}_i$ (the only paths for which the characteristic functions $\chi_{Q_i} (\omega_i )$ are nonzero),
\begin{align}
&\alpha \int_{\mathbb{R}} \d s \frac{e^{- |s|}}{2} \sum_{i, j = 1} ^N \int_0 ^T \frac{\d t}{| \omega_i (t) - \omega_j (t+s)|} - U \sum_{i < j} \int_0 ^T \frac{\d t}{| \omega_i (t) - \omega_j (t)|} \nonumber \\
&\leq  \sum_{k = 1 }^l \left( \alpha \int_{\mathbb{R}} \d s
\frac{e^{-|s|}}{2}
 \sum_{ i, j \in G_k } \int_0 ^T \frac{\d t}{|\omega_i (t) - \omega_j (t+s) |}
  - U \sum_{\substack{ i, j \in G_k \\ i<j }} \int_0 ^T \frac{\d t}{|w_i (t) - w_j (t) |}   \right) \nonumber \\
&+ \sum_{k_1 \neq k_2} T \sum_{\substack{i \in G_{k_1}\\ j \in
G_{k_2}}} \left(\frac{\alpha}{ d (i,j)}-\frac{U}{2d(i,j)+4 \sqrt{3}
R} \right).
 \label{eqn:Feynk}
\end{align}
The terms appearing in the first line of (\ref{eqn:Feynk}) involve
electrons either lying in the same cluster or lying in different
clusters; in the subsequent upper bound, the former have been left unchanged, but the latter have
been bounded above using what we know about the inter-particle
distance. To be precise, we have used that $d(i,j) \leq | \omega_i (t) - \omega_j (t+s ) |
\leq d(i,j) + 2 \sqrt{3} R$ for all $t$ and $s$.

   It follows that
\begin{align}
\frac{1}{T} \log Z_{Q_1 , ... , Q_N} (T) \leq \sum_{k = 1} ^l
\frac{1}{T} \log Z_{G_k} (T) - \sum_{k_1 \neq k_2	} \sum_{\substack{i
\in G_{k_1}\\ j \in G_{k_2}}} \left(\frac{\alpha}{
d(i,j)}-\frac{U}{2d(i,j)+4 \sqrt{3} R} \right), \label{eqn:fk2}
\end{align}
where, if $G_k = \{ i_1, ..., i_{N_k} \}$, then
\begin{align}
Z_{G_k} (T) := Z _{Q_{i_1} , ..., Q_{i_{N_k}} } (T),
\end{align}
and $Z _{Q_{i_1} , ..., Q_{i_{N_k}} } (T)$ is defined in the obvious
way by \eqref{eqn:feynz}. Taking the limit $T \to \infty$ on both
sides of (\ref{eqn:fk2}) yields \eqref{eqn:Feyngen}. Since $U \geq
0$, \eqref{eqn:Feyngencons} follows from \eqref{eqn:Feyngen}.
Solving for the $d$ that makes the terms appearing on the RHS of
(\ref{eqn:Feyngen})
 vanish and recalling property (P2) of the clusters we obtain (\ref{eqn:lem1}).
\end{proof}

{\bf Step 4} (Lower bound for cluster energies). Lemma \ref{lem:fk} tells us that we need only  bound each
$\widehat{E} (G_i)$ from below. To obtain such a bound, we will follow very closely \cite{LT} and \cite{MS}, and our methodology is essentially a straightforward generalization to the $N > 2$ case. We give here the proof for reader convenience, and also refer those interested in more details to these papers.

Our first step in deriving the lower bound will be to replace the
Hamiltonian $H_U^{(N_i)}$ with an ultraviolet cutoff Hamiltonian. For
$K_i>0$ (the value of our ultraviolet cutoff) let
\begin{align}
H^{N_i,U}_{K_i} = & \left( 1 - \frac{8 N_i \alpha}{\pi K_i} \right) \sum_{i=1} ^{N_i} p_i ^2 - 
\sqrt{ \alpha} \sum_{i=1} ^{N_i} \frac{1}{\sqrt{2}\pi} \int _{ |k| < K_i } \frac{\d k}{|k|} 
\left( e^{ikx_i} a(k) + e^{-ikx_i} a^* (k) \right) \nonumber \\
&+ \int_{|k|< K_i}\d k a^* (k) a (k) + \sum_{i <j} \frac{U}{| x_i -
x_j|}. \label{def:cutoff}
\end{align}
We now prove:
\begin{lemma} \label{lem:cutoff}
Let $H^{(N_i)}_U$ be the $N_i$-polaron Hamiltonian as in
(\ref{eqn:nham}) and $H^{N_i,U}_{K_i}$ be the cutoff Hamiltonian as in
(\ref{def:cutoff}). Then,
\begin{align}
H^{(N_i)}_U \geq H^{N_i,U}_{K_i} - \frac{1}{2}. \label{eqn:cutoffbd}
\end{align}
\end{lemma}
\begin{remark}
By completing the square, the phonon variables can be eliminated from $H^{N,U}_K$ as in \cite{LT}, for some cut-off $K$. This leaves only the positive kinetic energy, positive Coulomb energy and a negative constant proportional to $K$. Therefore, applying Lemma \ref{lem:cutoff} to $H_U^{(N)}$  shows that  $E_U^{(N)} (\alpha) > - \infty$.
\end{remark}
\begin{proof}
For notational simplicity we write $N$ and $K$ instead of $N_i$ and $K_i$, and relabel the electrons in the cluster $G_i$ as $1, ..., N$.
We define for each electron the vector operator
$\boldsymbol{Z^{(i)}} = (Z^{(i)}_1,Z^{(i)}_2,Z^{(i)}_3) $, where
\begin{align}
Z^{(i)}_j = \sqrt{\alpha} \int_{|k| \geq K} \d k \frac{k_j e^{ik \cdot
x_i}}{\sqrt{2} \pi |k|^3} a(k) .
\end{align}
A standard calculation yields
\begin{align}\label{eqn:commid}
\sum_{i=1}^N \sum_{j=1}^3 \left[ p_{ij} ,  Z^{(i)}_j - Z^{(i)*}_j
\right] = \sqrt{ \alpha} \sum_{i=1}^N \int_{ |k| \geq K }
\frac{\d k}{\sqrt{2} \pi |k|}\left[ e^{ik\cdot x_i} a(k) + e^{-ik
\cdot x_i} a^* (k) \right] =: H_I,
\end{align}
where the rightmost equality defines $H_I$. Above, $p_{ij}$ is the $i$-th electron momentum in direction $j$. We now bound the LHS of
(\ref{eqn:commid}). Fix $\epsilon  > 0.$ For arbitrary normalized
$\varphi$,
\begin{align}
\left| \sum_{j=1}^3 \langle \varphi\vert \left[ p_{ij} ,  Z^{(i)}_j - Z^{(i)*}_j \right] \vert \varphi \rangle \right| &
\leq 2 \langle \varphi \vert p_i ^2 \vert \varphi \rangle ^{1/2}\langle \varphi \vert- (\boldsymbol{Z^{(i)}} - \boldsymbol{Z^{(i) *}})^2\vert \varphi\rangle ^{1/2} \nonumber \\
& \leq 2 \langle \varphi \vert p_i ^2 \vert \varphi \rangle ^{1/2} \langle \varphi \vert 2(\boldsymbol{Z^{(i) *}} \boldsymbol{Z^{(i)}}+\boldsymbol{Z^{(i)}} \boldsymbol{Z^{(i)*}} ) \vert \varphi \rangle ^{1/2} \nonumber\\
& \leq \epsilon \langle \varphi \vert p_i ^2 \vert \varphi \rangle + \frac{2}{\epsilon} \langle \varphi \vert \boldsymbol{Z^{(i) *}} \boldsymbol{Z^{(i)}}+
\boldsymbol{Z^{(i)}} \boldsymbol{Z^{(i)*}}  \vert \varphi \rangle \nonumber \\
& =  \epsilon \langle \varphi \vert p_i ^2 \vert \varphi \rangle +
\frac{4}{\epsilon} \langle \varphi \vert \boldsymbol{Z^{(i) *}} \boldsymbol{Z^{(i)}} \vert
\varphi \rangle + \frac{1}{\epsilon} \frac{4 \alpha}{\pi  K}.
\label{eqn:commbd2}
\end{align}
In the last line we have used
\begin{align}
\sum_{j=1} ^3 \left[ Z^{(i)}_j, Z^{(i) *} _j \right] = \alpha
\int_{|k| \geq K} \frac{\d k}{2 \pi^2 |k|^4} = \frac{2 \alpha}{ \pi
K}.
\end{align}
   Additionally, the standard upper bound
\begin{align}
 \boldsymbol{Z^{(i) *}} \boldsymbol{Z^{(i)}} \leq \frac{2 \alpha}{\pi
K}  N_{\geq K} , \label{eqn:numbd}
\end{align}
holds, with
\begin{align}
N_{\geq K} = \int_{|k| \geq K} \d k a^* (k) a(k).
\end{align}

At this point, we take $\epsilon = 8 N \alpha / \pi K.$ With $H_I$
defined as in (\ref{eqn:commid}), the bounds (\ref{eqn:commbd2}) and
(\ref{eqn:numbd}) imply
\begin{align}
| \langle \varphi \vert H_I \vert \varphi \rangle | \leq \frac{8N
\alpha}{\pi K} \langle \varphi \vert \sum\nolimits_i p_i ^2\vert
\varphi \rangle + \langle \varphi \vert N_{\geq K} \vert \varphi
\rangle + \frac{1}{2} ,
\end{align}
from which it follows that
\begin{align}
H^{(N)}_U - H^{N,U}_K = \frac{8 N \alpha}{\pi K} \sum_i p_i ^2 +
N_{\geq K} - H_I \geq - \frac{1}{2}.
\end{align}
The above inequality proves Lemma \ref{lem:cutoff}.
\end{proof}

Let us focus on a single cluster $G_i$. We prove the following bound for the cluster energies:
\begin{proposition}\label{prop:cluster}
For the cluster $G_i$ let $P_i$ and $\delta_i$ be positive
constants, with $2\delta_i < 1$. Assume that the ultraviolet cut-off $K_i$ satisfies $8 \alpha N_i  < \pi K_i $. Then,
\begin{align}
\widehat{E} (G_i) \geq F(N_i) := \alpha^2
\mathcal{E}_{\widetilde{\nu_i}}^{(N_i)} (1) \frac{1}{(1 - 2\delta_i)^2 (1 -
\frac{8 \alpha N_i}{\pi K_i}) }
 - \frac{9 \alpha (R+d)^2 N_i^4 P_i^2 K_i}{ \pi  \delta_i} - (\frac{2K_i}{P_i}+1)^3 -\frac{1}{2} \label{eqn:clusterbound}
\end{align}
with $\widetilde{\nu_i} = U(1 - 2\delta_i ) / \alpha $.
\end{proposition}
\begin{remark}
The meaning of the constants $P_i$ and $\delta_i$ will be made clear below. 
The values of all the parameters introduced in Proposition \ref{prop:cluster} will be chosen later (each
claim of Theorem \ref{main} requires a different choice of constants). Our definition of $F(N_i)$ is suggestive; we are thinking of the lower bound for $\widehat{E} (G_i)$ as a function of the particle number $N_i$. Indeed, we will later see that the parameters $\delta_i, K_i$, etc. will depend only on $N_i$ (and $\alpha$, of course).
\end{remark}
\noindent \emph{Proof.} 
For notational simplicity of the proof we will omit the subscripts where the meaning is clear and write $N$ instead of $ N_i$, $K$ instead of $K_i$, etc.
In order to get a lower bound on $\widehat{E} (G_i)$ we fix an
arbitrary normalized wave function $\Psi$ with the required support
properties. By Lemma \ref{lem:cutoff} it suffices to consider the cut-off Hamiltonian as defined in (\ref{def:cutoff}). We will find a lower bound for $\langle \Psi \vert
H^{N,U}_K \vert \Psi \rangle$ which is independent of $\Psi$. Recall that our ultraviolet cut-off is denoted by $K>0$.
We will now split the sphere $B_K :=\{k \in \mathbb{R}^3: |k| \leq
K\}$ into cubes of length $P_i=P$ (whose value will be determined
later). For $ n = (n_1 , n_2 , n_3 ) \in \mathbb{Z}^3$ we define
$D_P (n) = [n_1 P - P/2, n_1 P + P/2] \times[n_2 P - P/2, n_2 P + P/2]
\times[n_3 P - P/2, n_3 P + P/2] \subset \mathbb{R}^3,$ the cube
with center $Pn$ and side length $P.$ Let $\Lambda _P = \left\{ n
\in\mathbb{Z}^3 : D_P (n) \cap B_K \neq  \emptyset \right\}$. For the
cardinality of $\Lambda_P$, denoted $| \Lambda _P |$, we have
\begin{align}
|\Lambda_P| \leq  (\frac{2K}{P}+1)^3.
\end{align}
For $n \in \Lambda_P$ set $B(n) = D_P (n) \cap B_K$, and let $k_{B(n)}$
be any point in $B(n).$ The electrons described by $\Psi$ are
localized in a cube of side length $N(R+d)$.
 By applying a unitary transformation (i.e., translating the system) we may assume that the cube is centered at the
origin. Then, using that the diameter of a cube is $\sqrt{3}$ times
its side length,  we obtain for $x_i$ in the support of $\Psi$
that
\begin{align}
|e^{ik \cdot x_i} - e^{ik_{B(n)} \cdot x_i}| \leq |k -
k_{B(n)}||x_i| \leq \frac{3P N}{2} (R+d). \label{eqn:quick1}
\end{align}
For $\delta=\delta_i > 0$, we have by completing the square,
\begin{align}
\sum_{n \in \Lambda_P}& \bigg[ - \sqrt{\alpha} \sum_{i=1}^N \int_{B(n)} \frac{\d k}{\sqrt{2}\pi |k|} \left( e^{i k \cdot x_i} - e^{i k_{B(n)} \cdot x_i} \right) a(k)+\left( e^{-i k \cdot x_i} - e^{-i k_{B(n)} \cdot x_i} \right) a^* (k)  \nonumber \\
&+ \delta \int_{B(n)} dk a^* (k) a(k)  \bigg] \nonumber \\
& \geq - \frac{\alpha N}{2 \pi ^2 \delta} \sum_{n \in \Lambda_P} \int_{B(n)} \frac{\d k}{|k|^2} \sum_{i=1}^N |e ^{ik x_i } -e^{i k_{B(n)} x_i} |^2 \geq -\frac{9  \alpha N^4 P^2  (R+d)^2 K}{2 \pi
\delta}. \label{eqn:loc2}
\end{align}
  Therefore,
\begin{align}
\langle \Psi \vert H_K ^{N,U} \vert \Psi \rangle \geq \langle \Psi
\vert \widehat{H}_{K} \left( \{ k_{B(n)}\} \right) \vert \Psi \rangle
-\frac{9  \alpha N^4 P^2 (R+d)^2 K}{2 \pi \delta},\label{eqn:quick2}
\end{align}
with
\begin{align}
& \widehat{H}_{K} \left( \{ k_{B(n)}\} \right) = ( 1 - \frac{8 \alpha N}{\pi K}  ) \sum_{i=1}^N p_i ^2 + U \sum_{i<j} \frac{ 1}{|x_i - x_j| } \nonumber \\
&+ \sum_{n \in \Lambda_P} \left[ (1- \delta) \int_{B(n)} \d k a^* (k)
a(k) - \sqrt{\alpha} \sum_{i =1}^N \int_{B(n)} \frac{\d k}{\sqrt{2}
\pi |k|} \left( e^{ik_{B(n)} \cdot x_i} a(k) + e^{-ik_{B(n)} \cdot
x_i} a^* (k) \right) \right].
\end{align}

      Define the block annihilation operator $A_n$ for $n \in \Lambda_P$ by
\begin{align}
A_n = \left( \int_{B(n)} \frac{\d k}{2 \pi^2 |k|^2} \right)
^{-1/2} \int_{B(n)} \frac{\d k}{\sqrt{2} \pi |k|} a(k).
\end{align}
 Each $A_n$ is a normalized boson mode satisfying $\left[ A_n, A_{n'}
^* \right] = \delta_{nn'}.$ By the Schwarz inequality, $A^* _n A_n
\leq \int\nolimits_{B(n)}\d k a^* (k) a(k) .$ Therefore $\widehat{H}_K
\left( \{ k_{B(n)} \} \right) \geq H^{Bl} _K \left( \{ k_{B(n)} \} 
\right)$, with the block Hamiltonian defined by
\begin{align}
H^{Bl}_K &  \left( \{ k_{B(n)} \} \right)  =(1 - \frac{8 \alpha N}{\pi K} ) \sum_{i=1} ^N p_i ^2 + U \sum_{i <j} \frac{1}{|x_i -x_j| } \nonumber \\
&+ \sum_{n \in \Lambda_P} \left[ (1- \delta) A^* _n A_n - \sqrt{
\alpha } \left( \int_{B(n)} \frac{\d k}{2 \pi^2 |k|^2
}\right)^{1/2}\sum_{i=1}^N \left( e^{i k_{B(n)}\cdot x_i} A_n +
e^{-ik_{B(n)} \cdot x_i}A^*_n \right) \right].\label{eqn:blockham}
\end{align}

        In summary, we have proven
\begin{align}
\langle \Psi \vert H_K ^{N,U} \vert \Psi \rangle \geq \langle \Psi
\vert H^{Bl}_K &  \left( \{ k_{B(n)} \} \right) \vert \Psi \rangle
-\frac{9  \alpha N^4 P^2 (R+d)^2 K}{2 \pi \delta}.\label{eqn:loc3}
\end{align}
     Note that the above inequality holds for any choice of $\{ k_{B(n)}
\}$. Our goal is now to find a lower
bound for $\langle \Psi \vert H_K^{Bl} ( \{ k_{B(n)} \} ) \vert \Psi \rangle$. A lower bound that is independent of $\Psi$ will in turn imply a lower bound for $\widehat{E}(G_i)$, by the above inequality and Lemma \ref{lem:cutoff}.

We will eventually use coherent states to relate the block Hamiltonian to the PT functional. Before doing so we would like to make one observation regarding
the block Hamiltonian. $H_K ^{Bl} (\{ k_{B(n)} \}) $  contains terms involving creation and annihilation operators of only finitely many functions. If we denote 
the span of these functions by $\mathcal{M}$, then we note that the bosonic Fock space $\mathcal{F} = \mathcal{F} ( L^2 (\mathbb{R}^3) )$ (where
$\mathcal{F} ( \mathcal H )$ denotes the bosonic Fock space over the Hilbert space $\mathcal H$) is unitarily equivalent to 
$\mathcal{F} ( \mathcal M ) \otimes \mathcal{F} ( \mathcal M ^\perp )$. Here, we recall
that $\mathcal{F} (H_1 \oplus H_2 ) \cong \mathcal{F} (H_1) \otimes \mathcal{F} (H_2)$.

Furthermore, because the block Hamiltonian contains no terms associated with the functions in $\mathcal M^\perp$, we can identify it with the
 operator $H^{Bl}_K (\{ k_{B(n)} \} ) \otimes \mathbb 1_{\mathcal F (\mathcal M ^\perp ) }$ which acts on the Hilbert space 
 $ L^2 (\mathbb R^{3N} ) \otimes \mathcal{F} ( \mathcal M) \otimes \mathcal F ( \mathcal M ^\perp )$. Here, we are thinking
  of $H_K^{Bl} ( \{ k_{B(n)} \} )$ defined by (\ref{eqn:blockham}) as an operator acting on 
  $L^2 (\mathbb R^{3N} ) \otimes \mathcal F ( \mathcal M ) $. Since the operators $A \otimes \1$ and $A$ have the same ground state energy, it is therefore sufficient to obtain a lower bound for the ground state energy of  $H_K^{Bl} ( \{ k_{B(n)} \} )$ acting on $L^2 (\R^{3N} ) \otimes \F (\M)$,  for a choice of $\{k_{B(n)} \} = \{k(1)_n\}$ which we now fix. Moreover, since the unitary transformation taking $L^2 (\R ^{3N} ) \otimes \F$ to $L^2 (\R^{3N} ) \otimes \F (\M) \otimes \F (\M^\perp )$ acts only on the phonon variables, we can still assume that the wave functions that $H_K ^{Bl} ( \{k(1)_n\})$ act on have the same support properties as before - that is, they are localized in a box at the origin.

We now fix an arbitrary wave function in $L^2 (\R^{3N} ) \otimes \F (\M)$ which we continue to denote by $\Psi$ which has the required support properties. We claim that instead of finding a lower bound for the energy of $\Psi$ using the $\{k(1)_n\}$ we fixed previously, we can instead take the supremum of $\langle \Psi \vert H^{Bl}_K ( \{ k_{B(n)} \} ) \vert \Psi \rangle$ over all choices $\{k_{B(n)}\}$, at the cost of replacing $(1 - \delta)$ by $(1 - 2 \delta)$ in the definition of $H^{Bl}_K (\{ k_{B(n)} \})$ and incurring an error term which is equal to the constant appearing in (\ref{eqn:loc3}). What we mean by this is the following. Let $\{ k_{B(n)} \}$ be another choice of the fixed points in each $B(n)$. The arguments of (\ref{eqn:quick1}) - (\ref{eqn:quick2}), applied to the block oscillators $A^\#(k)$ instead of the $a^\#(k)$, show that
\begin{align}
\langle \Psi \vert H_K^{Bl} ( \{k(1)_n\}) \vert \Psi \rangle \geq \langle \Psi \vert H_K^{Bl} (\{ k_{B(n)} \} ) \vert \Psi \rangle - \delta \sum_{n \in \Lambda_P } \langle \Psi \vert A^*_n A_n \vert \Psi \rangle - \frac{9  \alpha N^4 P^2 (R+d)^2 K}{2 \pi \delta}. \label{eqn:last2}
\end{align}
We can even take the supremum on the RHS of (\ref{eqn:last2}) over choices of $\{k_{B(n)} \}$. Define $\widetilde{H}_K^{Bl} (\{ k_{B(n)} \})$ by the RHS of (\ref{eqn:blockham}) except with the coefficient $(1 - 2\delta)$ in front of the block number operator instead of $(1 - \delta)$. We will now prove the bound
\begin{align}
\sup_{ \{ k_{B(n)} \} } \langle \Psi \vert \widetilde{H}_K^{Bl} (\{ k_{B(n)} \} ) \vert \Psi \rangle \geq \frac{\alpha^2
\mathcal{E}_{\widetilde{\nu}}^{(N)} (1) }{(1 -2 \delta)^2 (1 -
\frac{8 \alpha N}{\pi K}) } - (\frac{2K}{P}+1)^3  \label{eqn:last1}
\end{align}
Let us see that such a lower bound implies Proposition \ref{prop:cluster}. Since the RHS of (\ref{eqn:last1}) does not depend on $\Psi$, this, together with (\ref{eqn:last2}), implies a lower bound for the ground state energy of $H_K^{Bl} ( \{k(1)_n\} ) $. Recall that this operator is acting only on wave functions in the Hilbert space $L^2 (\R) \otimes\F (\M)$ that have the aforementioned support properties. As we have previously argued, this implies a lower bound for $\widehat{E} (G_i)$. Collecting (\ref{eqn:cutoffbd}), (\ref{eqn:loc3}), (\ref{eqn:last2}) and (\ref{eqn:last1}) we see that the bound we obtain is in fact
\begin{align}
\widehat{E} (G_i) \geq F(N) ,
\end{align}
with $F(N)$ defined as in (\ref{eqn:clusterbound}), which proves Proposition \ref{prop:cluster}.

We have left only to prove (\ref{eqn:last1}). For a given $\xi=(\xi_n)_{n \in \Lambda_P} \in \mathbb C ^{| \Lambda_P |}$ we introduce the
coherent state in $\mathcal{F}(\mathcal M)$ defined by
\begin{equation}
| \xi \rangle := \prod_{n \in \Lambda_P} e^{\xi_n A_n^*-
\overline{\xi}_n A_n} \Omega_{\mathcal M} ,
\end{equation}
with $\Omega_{\mathcal M}$ the vacuum state in $\mathcal F (\mathcal M)$. With this definition, we have that
\begin{equation}\label{eqn:intid}
\int \d \xi \d \xi^* | \xi \rangle \langle \xi|= \mathbb 1,\quad \int \d \xi \d \xi ^* \xi_n \vert \xi \rangle \langle \xi \vert = A_n ,
 \quad \int \d \xi \d \xi ^* ( |\xi_n |^2 -1 ) \vert \xi \rangle \langle \xi \vert=A^* _n A_n,
\end{equation}
where $\d \xi \d \xi ^*= \prod_{n \in \Lambda_P} \pi ^{-1} \d \xi_n \d\xi_n ^*$ and $\d \xi_n \d \xi_n ^* $ denotes Lebesgue measure. 

Let $\Psi_\xi:= \langle  \Psi, \xi \rangle_{\mathcal{F}(\mathcal{M})} \in 
L^2(\mathbb{R}^{3N})$ 
  be the wave function obtained by integrating out the phonon variables. With the identities (\ref{eqn:intid}) we have,
 \begin{align}
\langle \Psi \vert \widetilde{H}_K^{Bl}(\{ k_{B(n)}\}) \vert \Psi \rangle=\int \d \xi \d \xi^* \langle
\Psi_\xi \vert h_\xi \left( \{ k_{B(n)} \} \right) \Psi_\xi
\rangle_{\LM}\label{eqn:fin2},
 \end{align}
where
\begin{align}
h_{\xi}& \left( \{ k_{B(n)} \} \right) = (1- \frac{8 \alpha N}{\pi K} ) \sum_{i=1} ^N p_i ^2 + U \sum_{i<j} \frac{1}{|x_i- x_j |} \nonumber \\
&+\sum_{n \in \Lambda_P } \left[( 1- 2\delta) (|\xi_n|^2 -1 )
-\sqrt{\alpha} \left( \int_{B(n)} \frac{dk}{2 \pi^2 |k|^2} \right)
^{1/2} \sum_{i=1}^N \left(\xi_n e^{i k_{B(n)}\cdot x_i} + \xi_n ^*
e^{-i k_{B(n)} \cdot x_i} \right)  \right].
\end{align}
 Completing the square we obtain
\begin{align}
& \langle \Psi_{\xi} \vert h_{\xi} \left( \{k_{B(n)} \} \right) \vert \Psi_{\xi} \rangle \geq \langle \Psi_{\xi}\vert   (1- \frac{8 \alpha N}{\pi K} )\sum_{i=1} ^N p_i ^2 + U \sum_{i<j} \frac{1}{|x_i- x_j |} \vert \Psi_{\xi} \rangle \nonumber \\
&- \sum_{n \in \Lambda_P} \left[\frac{4 \pi \alpha}{1 -2 \delta}\int_{B(n)} \d k \frac{| \hat{\rho}_{\xi} (k_{B(n)}) |^2}{|k|^2 \norm{\Psi_{\xi}}^2_{\LM}}  \right] - (1 - 2\delta) \norm{\Psi_{\xi}}^2_{\LM} |\Lambda_P |,
\end{align}
where $\hat{\rho}_{\xi} (k) =(2 \pi)^{-3/2} \int\nolimits
_{\mathbb{R}^3}dx e^{ik \cdot x} \rho _{\xi} (x) $ is the Fourier
transform of
\begin{align}
\rho_{\xi} (x) = \sum_{i=1}^N \int_{\mathbb{R}^{3(N-1)}} |
\Psi_{\xi} (x_1, ..., x_{i-1}, x, x_{i+1}, ..., x_M ) |^2 \d x_1 ...
\widehat{\d x }_i ... \d x_N.
\end{align}
 We have the inequality
\begin{align}
\inf_{\{k_{B(n)} \}}\int_{B(n)} \d k \frac{| \hat{\rho}_{\xi} (k_{B(n)}) |^2}{|k|^2 \norm{\Psi_{\xi}}^2_{\LM}} \leq \int_{B(n)} \d k \frac{| \hat{\rho}_{\xi} (k) |^2}{|k|^2 \norm{\Psi_{\xi}}^2_{\LM}}
\end{align}
 and so,
\begin{align}
& \sup_{\{k_{B(n)} \}} \int \d \xi \d \xi ^* \langle \Psi_{\xi} \vert h_{\xi} \left( \{ k_{B(n)} \} \right) \vert \Psi_{\xi} \rangle_{\LM} \nonumber \\
\geq& \int \bigg\{\langle \Psi_{\xi} \vert (1- \frac{8 \alpha N}{\pi K} )\sum_{i=1} ^N p_i ^2 + U \sum_{i<j} \frac{1}{|x_i- x_j |} \vert \Psi_{\xi} \rangle \nonumber \\
&- \frac{4\pi \alpha}{1 - 2\delta}\int_{\mathbb{R}^3} \d k \frac{| \hat{\rho}_{\xi} (k) |^2}{|k|^2 \norm{\Psi_{\xi}}^2_{\LM}}  - (1 - 2\delta) \norm{\Psi_{\xi}}^2_{\LM} |\Lambda_P | \bigg\} \d\xi \d\xi ^*  \nonumber \\
=& \int  \d \xi \d \xi ^*  \norm{\Psi_{\xi}}^2_{\LM} (1- \frac{8 \alpha N}{\pi K} ) \bigg[\langle \widetilde{\Psi}_{\xi} \vert \sum_{i=1} ^N p_i ^2 + \frac{U}{1-  \frac{8 \alpha N}{\pi K} } \sum_{i<j} \frac{1}{|x_i- x_j |} \vert \widetilde{\Psi}_{\xi} \rangle\nonumber \\
&-\frac{\alpha}{ (1-2\delta)(1-  \frac{8 \alpha N}{\pi K})}\frac{1}{
\norm{\Psi_{\xi}}^4_{\LM}} \int \int_{\mathbb{R}^3
\times \mathbb{R}^3} \d x \d y \frac{\rho_{\xi} (x) \rho_{\xi}
(y)}{|x-y|}\bigg] -(1-2\delta)|\Lambda _P | \label{eqn:loc4}
\end{align}
for $\widetilde{\Psi}_{\xi} = \Psi_{\xi} / \norm{\Psi_{\xi}}_{\LM}.$ Above, we have used,
\begin{align}
\int_{\mathbb{R}^3} \d k \frac{\bar{\hat{f}}(k) \hat{g} (k)}{|k|^2} = \frac{1}{4 \pi}
\int \int_{\mathbb{R}^3 \times \mathbb{R}^3} \d x \d y \frac{\bar{f}(x)
g(y)}{|x-y|}.
\end{align}

The integrand enclosed in square brackets in the last two lines of
(\ref{eqn:loc4}) is just a constant times the PT functional for
$\widetilde{\Psi}_{\xi}$ with the coefficient $\alpha$ replaced by
$\alpha (1-  \frac{8 \alpha N}{\pi K} )^{-1} (1- 2 \delta )^{-1}$ and
$U$ replaced by $U  (1-  \frac{8 \alpha N}{\pi K} )^{-1}.$ 
Using this and applying
 (\ref{eqn:scaling}) we obtain the lower bound,
\begin{align}
&\langle \widetilde{\Psi}_{\xi} \vert \sum_{i=1} ^N p_i ^2 +
\frac{U}{1- \frac{8 \alpha N}{\pi K} } \sum_{i<j} \frac{1}{|x_i- x_j
|} \vert \widetilde{\Psi}_{\xi} \rangle  \nonumber \\
&- \frac{\alpha}{ (1- \frac{8
a N}{\pi K} ) (1-2\delta)} \frac{1}{ \norm{\Psi_{\xi}}^4_{\LM}} 
\int \int_{\mathbb{R}^3 \times \mathbb{R}^3} \d x \d y \frac{\rho_{\xi} (x) \rho_{\xi} (y)}{|x-y|} 
\geq \mathcal{E} ^{(N)} _{\widetilde{\nu}} (1) \frac{ \alpha ^2}{(1 -
\frac{8 \alpha N}{\pi K} )^2 (1 - 2\delta)^2},\label{eqn:fin1}
\end{align}
for $\widetilde{\nu} = U (1 - 2\delta ) / \alpha .$ The equality (\ref{eqn:fin2}) together with the bounds
(\ref{eqn:loc4}) and (\ref{eqn:fin1}) prove (\ref{eqn:last1}), thus completing the proof of Proposition \ref{prop:cluster}. \qed

{\bf Step 5} (Choice of parameters) Let us summarize the proof so
far. In Proposition \ref{prop:localization}, we proved that the
ground state energy $E^{(N)}_U ( \alpha )$ is bounded below by the infimum of
$\langle \Phi  \vert H^{(N)}_U \vert \Phi \rangle$ taken over wave
functions in $\mathcal{B}_R$ (the set of normalized wave functions with
electrons localized in boxes somewhere in space) minus an error term.
Subsequently, for an arbitrary wave function $\Phi$ in
$\mathcal{B}_R$ we separated the electrons into the clusters
$\{G_i \}_i$ and have shown in Lemma \ref{lem:fk} that the energy of $\Phi$ can
be bounded below by a sum of cluster energies, $\sum_i \widehat{E} (G_i )$,
minus an error term. In Proposition \ref{prop:cluster} we derived a
lower bound on each $\widehat{E} (G_i) $. It remains only to collect all of these bounds and choose values for the parameters (e.g., the $K_i$, $R$, etc.) that we have
introduced in the proof. We will subsequently obtain a lower bound for the ground state
energy.

 In the case $U > 2 \alpha$ we take $ d = \sqrt{3} R / (\nu -2 )$ and using equations \eqref{eqn:lowerb},
\eqref{eqn:lem1} and \eqref{eqn:clusterbound} we obtain that (recalling the definition of $F(N_i)$),
\begin{equation}\label{clusterreduc}
E_U^{(N)}(\alpha) \geq \min \sum_{i=1}^l \left(
F(N_i)-\frac{3 N_i \pi^2}{R^2}\right),
\end{equation}
where the minimum is taken over all the possible choices of positive
integers $N_1,...,N_l$, with $\sum_{i=1}^l N_i=N$, for any $l$. By taking the minimum, we have removed all reference to the original wave function $\Phi$, and therefore obtain a lower bound for the ground state energy. Recall that we are thinking of $F(N_i)$ as a function of the particle number $N_i$, so that the constants on the RHS of (\ref{eqn:clusterbound}) are thought of as depending on how many particles are in the cluster.

We will now make our choice of the parameters $K_i, R$, etc., so that (\ref{clusterreduc}) yields the claim (\ref{eqn:main1}) of Theorem \ref{main}. 
 Let us first explain the reasoning behind our first choice of parameters. In the $U > 2 \alpha$ regime, the Pekar-Tomasevich energy $\mathcal{E} ^{(N)}_U (\alpha )$ is of order  $N$
 \cite{FLST}.  As a consequence, the multiplicative errors appearing in (\ref{eqn:clusterbound}) (i.e., the $1 - 2\delta $ and $1-\frac{8 \alpha N}{ \pi K}$ terms) result in corrections to
the leading order term $\alpha^2 \mathcal{E}^{(N)}_\nu (1)$  which are of the order of $\alpha^2 N \delta $ and  $ \alpha^2 N \frac{8 \alpha N}{ \pi K}$. Therefore, due to \eqref{eqn:clusterbound}, the corrections contained in $F(N_i)-\frac{3 N_i \pi^2}{R^2}$ to the leading PT energy term are (up to constants),
\begin{align}
\frac{\alpha^3 N_i^2}{ K_i},\quad \alpha^2 N_i \delta_i , \quad\frac{N_i}{R^2}, \quad \frac{\alpha N_i^4 P_i^2 R^2 K_i}{\delta_i},\quad \frac{K_i^3}{ P_i^3 }.
\end{align}
In order to optimize our bounds, we choose the parameters so that the powers of $N_i$ and $\alpha$ appearing in the above terms are all equal. A simple calculation shows that we should take $K_i =
N_i ^{1/5} \alpha^{6/5}c_1$,
 $R = c_2\alpha^{-9/10}
N^{-2/5} $, $P_i = c_3 \alpha^{3/5} N_i ^{-2/5}$ and $\delta_i = c_4
\alpha^{-1/5}N_i ^{4/5} $ where $c_1,c_2,c_3,c_4 >0$ are constants. 
In order to ensure that the hypotheses on $\delta_i$ and $K_i$ stated
 in Proposition \ref{prop:cluster} are satisfied, we choose the constants to satisfy $\alpha > N^4 \times \max\{ ( 2\nu c_4 / (\nu - 2))^5 , (8/(\pi c_1))^5 \} $. There is a subtle point here. The leading order term in our lower bound is $\alpha ^2 \E ^{(N_i)} _{\widetilde \nu } (1)$, with $\widetilde \nu_i = (1 -2 \delta_i ) \nu$, which is slightly smaller than  the original $\nu$. The assumption that $\alpha > N^4 (2 \nu c_4 / (\nu - 2))^5$ ensures that $\widetilde \nu_i > 2$ so that $\E ^{(N_i)} _{\widetilde \nu } (1)$ is order $N_i$.

With this choice of constants, it follows that
\begin{align}
\mbox{RHS of } (\ref{clusterreduc}) &\geq \frac{ \alpha^2
\mathcal{E}_{\widetilde{\nu}}^{(N)} (1)}{(1 - \frac{2 c_4 N^{4/5}}{\alpha^{1/5}})^2 (1 -
\frac{8 N^{4/5}}{c_1 \pi \alpha^{1/5}}) }  -   \alpha^{9/5}  N^{9/5}\frac{9(\nu +\sqrt{3}-2)^2c_1c_2^2c_3^2 }{\pi c_4 (\nu  -2)^2}\nonumber \\
& - \left( \alpha^{3/5}N^{3/5}\frac{2 c_1 }{c_3}+1\right )^3 - \alpha^{9/5} N^{9/5}\frac{2 \pi^2}{c_2 ^2} - N\left(\frac{3}{2} + \frac{6 \alpha^{3/5}c_1}{c_3} \right) \label{eqn:finalbd}
\end{align} 
with $\widetilde \nu = (1 - 2 c_4 \alpha^{-1/5} N^{4/5})\nu$. Here we have used subaddivity of the PT ground state energy and also the fact that the PT ground state energy is negative. We have also used the elementary inequality $\sum_i N_i ^p \leq N^p$ for $p \geq 1$. Note that because it is a lower bound for the ground state energy, (\ref{eqn:finalbd}) alone confirms the asymptotic formula (\ref{eqn:limit}).

We are left with showing that the RHS of (\ref{eqn:main1}) is a lower bound for the RHS of (\ref{eqn:finalbd}). Clearly, (\ref{eqn:finalbd}) is of the same form as (\ref{eqn:main1}), except for the multiplicative errors in the leading order PT energy term. The terms appearing on the denominator of the leading order  term obviously result in corrections that are order $N^{9/5} \alpha^{9/5}$, leaving us with a lower   bound of form $\alpha^2 \E_{\widetilde \nu} ^{(N)} (1) - C N^{9/5} \alpha^{9/5}$. The PT energy    $ \E_{\widetilde \nu} ^{(N)} (1)$ (with a weaker Coulomb repulsion) can be bounded below by $ \E_{ \nu} ^{(N)} (1)$ (with the original stronger Coulomb repulsion) minus an error term which is order $(\nu - \widetilde \nu ) \E_{\nu} ^{(N)}(1) $ using the concavity of the function $\nu \to \mathcal E ^{(N)} _\nu (1) $.\footnotemark{\footnotetext{The concavity is due to the fact that the 
   ground state energy is, by definition, an infimum of affine functions of $\nu$.}} This proves the first claim of Theorem \ref{main}.

For general $\nu>0$ the Pekar-Tomasevich energy $\mathcal{E}_U^{(N)}(\alpha)$ is always at most of the order $ N^3$ \cite{BB}. To derive (\ref{eqn:main2}), we repeat the above argument, but use \eqref{eqn:Feyngencons} instead of \eqref{eqn:lem1} to bound $E^{(N)}_U (\alpha)$ by the cluster energies $\{ \widehat E (G_i) \}_i$. We choose for our parameters  $K_i = N_i \alpha^{27/23} c_1$, $R = c_2\alpha^{-19/23} N^{-1}$, $d=c_5 R$, $P_i =c_3 \alpha^{13/23} $ and
$\delta_i =c_4 \alpha^{-4/23} $. With these choices, we obtain an inequality similar to (\ref{eqn:finalbd}) except for the fact that $N^{9/5}$'s have become $N^3$'s and that there is an additional term arising from our application of the Feynman-Kac formula (i.e., the extra term on the RHS of (\ref{eqn:Feyngencons})) which contains a higher power of $\alpha$, namely $\alpha^{42/23}$. The same argument as that above then yields the RHS of (\ref{eqn:main2}) as a lower bound for $E^{(N)}_U (\alpha)$.  Note that we require $\alpha > \max \{(2c_4)^{23/4}, (8/(\pi c_1))^{23/4}\}$ in order to satisfy the hypotheses of Proposition \ref{prop:cluster}.  This completes the proof of Theorem \ref{main}.


\begin{thebibliography}{9}
\bibitem[AD]{AD} A.S. Alexandrov and J.T. Devreese, \emph{Advances in Polaron Physics,} Springer, Berlin, 2010.
\bibitem[BB]{BB} R.D. Benguria, G. A. Bley, Exact asymptotic behavior of the
Pekar-Tomasevich functional, \emph{ J. Math. Phys.,} {\bf 52} no. 5, (2011), 052110.
\bibitem[DV]{DV}  M.D. Donsker, S.R.S. Varadhan, Asymptotics for the
polaron, \emph{Comm. Pure Appl. Math.,} {\bf 36} no. 4, (1983), 505-528.
\bibitem[Fr]{Fr} H. Fr\"ohlich, Theory of electrical breakdown in ionic crystals, \emph{Proc. R. Soc. Lond. A}, {\bf 160} (1937), 230-241.
\bibitem[FLST]{FLST} R. Frank, E.H. Lieb, R. Seiringer, L.E. Thomas, Stability and absence of binding for multi-polaron
systems, \emph{ Publ. Math. IHES} {\bf 113}  no. 1, (2011), 39 - 67.
\bibitem[GW]{GW} M. Griesemer, D. Wellig, Strong coupling polarons in electromagnetic fields, in preparation.
\bibitem[L]{L} E.H. Lieb, Existence and uniqueness of the minimizing solution of Choquard's nonlinear equation. \emph{Studies in Appl. Math}. {\bf 57}, (1977), 93-105.
\bibitem[LT]{LT}E.H. Lieb, L.E. Thomas, Exact ground state energy of the strong-coupling
polaron,  \emph{Commun. Math. Phys.,} {\bf 183} (1997), 511-519. Erratum: ibid., {\bf 188} (1997), 499-500.
\bibitem[Le]{Le} M. Lewin, Methods for nonlinear many-body quantum systems. \emph{J. Funct. Anal}, {\bf 260} (2011) 3535-3595.
\bibitem[M]{M}O. Madelung, \textit{Introduction to Solid-State Theory}. Springer (1981).
\bibitem[MS]{MS} T. Miyao, H. Spohn, The bipolaron in the strong coupling limit, \emph{Ann. Henri Poincar\'{e},} {\bf 8} no. 7, (2007), 1333-1370.
\bibitem[N]{N} E. Nelson, Interaction of non-relativistic particles with a quantized scalar field. \emph{J. Math. Phys.},{\bf 5} (1964), 1190-1197.
\bibitem[P]{P} S.I. Pekar, \emph{Untersuchung \"uber die Elektronentheorie der Kristalle.} Akademie Verlag, Berlin, 1954.
\bibitem[R]{R} G. Roepstorff, \emph{Path Integral Approach to Quantum Physics}, Springer, Berlin-Heidelberg-New York, 1994.

\end{thebibliography}
\end{document}